\documentclass[letterpaper, 10 pt, conference]{ieeeconf}  

\IEEEoverridecommandlockouts                              

\usepackage{mathtools}
\usepackage{amsmath}
\usepackage{amssymb}
\usepackage{mathrsfs}
\usepackage{graphicx}
\usepackage{subfigure}
\usepackage{color}
\usepackage{bm}
\usepackage{array}
\usepackage{algorithm}
\usepackage{algorithmic}

\usepackage{savesym}
\savesymbol{labelindent}
\usepackage{enumitem}

\makeatletter
\let\NAT@parse\undefined
\makeatother
\usepackage[hidelinks]{hyperref}

\allowdisplaybreaks[4]

\newtheorem{asmp}{\bf\emph{Assumption}}
\newtheorem{lem}{\bf\emph{Lemma}}
\newtheorem{prop}{\bf\emph{Proposition}}
\newtheorem{thm}{\bf\emph{Theorem}}

\newtheorem{rmk}{\bf\emph{Remark}}
\newtheorem{defn}{\bf\emph{Definition}}

\allowdisplaybreaks[4]

\newcommand{\RR}{\mathbb{R}}

\newcommand{\NN}{\mathbb{N}}

\newcommand{\eps}{\varepsilon}

\newcommand{\mtci}{\mathcal{I}}

\newcommand{\mtcl}{\mathcal{L}}

\newcommand{\mtcq}{\mathcal{Q}}

\newcommand{\mtcs}{\mathcal{S}}
\newcommand{\mtct}{\mathcal{T}}

\newcommand{\mtcv}{\mathcal{V}}

\newcommand{\mtcx}{\mathcal{X}}

\newcommand{\bfy}{\mathbf{y}}

\newcommand{\Let}{\coloneqq}

\newcommand{\tx}{\textup}
\newcommand{\tp}{\top}

\newcommand{\diag}{\tx{diag}}

\newcommand{\ab}{\tx{ab}}
\newcommand{\dup}{\tx{dup}}
\newcommand{\mtctab}{\mtct^{\ab}}
\newcommand{\thetaab}{\theta^{\ab}}
\newcommand{\Xab}{X^{\ab}}
\newcommand{\mtcXab}{\mtcx^{\ab}}

\newcommand{\toab}{\to^{\ab}}
\newcommand{\Yab}{Y^{\ab}}
\newcommand{\yab}{y^{\ab}}
\newcommand{\Hab}{H^{\ab}}
\newcommand{\Lambdaab}{\Lambda^{\ab}}
\newcommand{\lambdaab}{\lambda^{\ab}}
\newcommand{\Wab}{W^{\ab}}
\newcommand{\Thetaab}{\Theta^{\ab}}
\newcommand{\xab}{x^{\ab}}
\newcommand{\xchkab}{\check{x}^{\ab}}
\newcommand{\txH}{\tx{H}}
\newcommand{\xchk}{\check{x}}

\newcommand{\dlambda}{d_{\lambda}}
\newcommand{\dx}{d_{x}}

\newcommand{\init}{\tx{init}}

\usepackage{xcolor}
\definecolor{colorMS}{rgb}{0.35,0.4,0.6}

\title{\LARGE \bf
Approximate Simulation-Based Verification of Compatibility \\of the Friedkin--Johnsen Model with Binary Observations*
}
\author{Yu Xing$^{1}$, Aneesh Raghavan$^{2}$, Michael T. Schaub$^{1}$, and  Karl H. Johansson$^{2}$
\thanks{*This work was supported by the Alexander von Humboldt Foundation, Knut and Alice Wallenberg
Foundation (Wallenberg Scholar grant), the Swedish Research Council (Distinguished Professor grant 2017-01078), and the Swedish Foundation for Strategic
Research (SUCCESS FUS21-0026).}
\thanks{$^{1}$YX and MTS are with the Faculty of Computer Science, RWTH Aachen University, Germany. MTS is also with the Department of Biology, RWTH Aachen University.
        {\tt\small \{yu.xing, michael.schaub\}@rwth-aachen.de}}%
\thanks{$^{2}$AR and KHJ are with Department of Decision and Control Systems, School of Electrical Engineering and Computer Science, KTH Royal Institute of
Technology, and also with Digital Futures, Stockholm, Sweden.
        {\tt\small \{aneesh,kallej\}@kth.se}}%
}

\begin{document}

\maketitle
\thispagestyle{empty}
\pagestyle{empty}

\begin{abstract}
We consider a verification problem for opinion dynamics based on binary observations.
The opinion dynamics is governed by a Friedkin-Johnsen (FJ) model, where only a sequence of binary outputs is available instead of the agents' continuous opinions. 
At every time-step we observe a binarized output for each agent depending on whether the opinion exceeds a fixed threshold.
The objective is to verify whether an FJ model with a given set of stubbornness parameters and initial opinions can generate the observed binary outputs up to a small error.
The FJ model is formulated as a transition system, and an approximate simulation relation of two transition systems is defined in terms of the proximity of their opinion trajectories and output sequences.
We then construct a finite set of abstract FJ models by simplifying the influence matrix and discretizing the stubbornness parameters and the initial opinions.
It is shown that the abstraction approximately simulates any concrete FJ model with continuous parameters and initial opinions, and is itself approximately simulated by some concrete FJ model. 
These results ensure that consistency verification can be performed over the finite abstraction. 
Specifically, by checking whether an abstract model satisfies the observation constraints, we can conclude whether the corresponding family of concrete FJ models is consistent with the binary observations.
Finally, numerical experiments are presented to illustrate the proposed verification framework.
\end{abstract}

\section{Introduction}
The study of social opinion dynamics has attracted significant interest across disciplines~\cite{castellano2009statistical}, due to its wide range of applications, such as in marketing and recommender systems~\cite{zha2020opinion}.
In recent years, an increasing number of studies have focused on learning such dynamics from data~\cite{ravazzi2021learning,schaub2020blind,wai2022community}.
However, existing methods often require collecting a long sequence of data from a single trajectory or snapshots from multiple trajectories.
In practice, typically only a few opinion samples are available for each individual, and the data may be discrete-valued rather than continuous~\cite{effrosynidis2022climate,okawa2022predicting,kozitsin2022formal}.
It is thus important to investigate how to infer information about social dynamics from such limited data.

\subsection{Related Work}
Opinion dynamics with continuous states include linear averaging models such as the DeGroot model and the Friedkin--Johnsen (FJ) model~\cite{friedkin1990social}.
With the rise of online social network platforms, researchers have increasingly investigated learning algorithms for opinion dynamics~\cite{ravazzi2021learning}. 
Learning sparse networks based on steady states has been studied in~\cite{wai2016active} for the DeGroot model
, and in~\cite{ravazzi2017learning} for the FJ model.
Learning network structure from discrete-valued dynamical observations has been studied in e.g.,~\cite{xing2023recursive,xie2023finite}.
Inferring community structures is another approach to reducing sample complexity for large networked dynamics, such as epidemics~\cite{peixoto2019network}, cascade dynamics~\cite{prokhorenkova2022less}, and opinion dynamics~\cite{xing2024learning}.
A blind community detection approach is proposed in~\cite{schaub2020blind,wai2022community}, which applies spectral clustering to sample covariance matrices.

Data for learning opinion dynamics can take various forms. 
Small-group controlled experiments can generate complete opinion evolution trajectories~\cite{friedkin1990social,vande2016modelling}. 
However, data collected from online discussions usually contain only a few samples for each individual~\cite{effrosynidis2022climate,okawa2022predicting,kozitsin2022formal}, which is typically much smaller than the network size.
Moreover, although opinions can be quantified by methods such as opinion leaning estimation~\cite{cinelli2021echo} and Bayesian inference~\cite{barbera2015birds}, stance expressed in user posts is often coarsely classified into a few discrete categories (e.g., favor, against, neutral)~\cite{effrosynidis2022climate,mohammad2016semeval}.
These challenges indicate the need to develop new methods capable of inferring models about opinion dynamics from limited and discrete data.

Formal methods~\cite{baier2008principles,tabuada2009verification} have widely been applied in control systems, providing a rigorous framework for expressing and reasoning about complex system properties.
Model checking aims to verify whether a system satisfies given specifications~\cite{baier2008principles}, whereas control synthesis focuses on designing controllers that drive the system toward desired behaviors~\cite{yordanov2011temporal}.
Several works have extended formal verification to swarm systems.
For example,~\cite{konur2010formal} introduces a probabilistic state transition system to verify global behavior of robot swarms.
For an open multi-agent system describing diffusion dynamics,~\cite{belardinelli2015formal} shows that verification problems can be reduced to checking finite abstractions.
The paper~\cite{kouvaros2015verifying} develops a framework to verify temporal-epistemic properties of swarm systems, and~\cite{kouvaros2016formal} applies it to verifying consensus properties in opinion dynamics.

\subsection{Contribution}
In this paper we consider a verification problem for the FJ model with binary observations.
It is assumed that agents generate binarized outputs at each time step. That is, their continuous opinions cannot be observed.
The problem of interest is to determine whether there exists an FJ model, with a given set of parameters, that can generate the observed binary observations.

The FJ model is remodeled as a transition system on a finite set by discretizing the system parameters and initial conditions. 
Subsequently the model identification problem is formulated as a verification problem using temporal logic, and an approximate simulation relation is introduced suited for the discontinuous quantization function.
Unlike existing approaches, the proposed approximate simulation relation introduced in this paper incorporates the proximity of both opinions and outputs.
It is shown that the constructed abstraction can approximately simulate the original FJ dynamics when the discretization is sufficiently fine (Theorem~\ref{thm:apprx_simul}).
As a consequence, the abstraction yields a necessary condition for consistency with observations (Proposition~\ref{prop:verification}): 
Violation of the specification by the abstraction guarantees violation by the original system, whereas satisfaction by the abstraction indicates potential consistency.

Unlike prior work that investigates verification of global behaviors in swarm systems~\cite{belardinelli2015formal,kouvaros2015verifying,kouvaros2016formal}, the current work shifts the focus to reasoning about system parameters based on observed behaviors, providing a novel approach to understanding complex dynamics from limited data.
By leveraging logical specifications, the proposed method can incorporate missing or partial data as constraints and allows individual information, providing a flexible framework for reasoning about multi-agent systems.

\subsection{Outline}
The rest of the paper is organized as follows.
In Section~\ref{sec:preliminaries} we introduce preliminary knowledge and formulate the problem.
Section~\ref{sec:result} provides main results.
Section~\ref{sec:simul} presents numerical experiments, and Section~\ref{sec:conclusion} concludes the paper.

\subsubsection*{Notation} 
Let $\mathbb{R}^n$, $\mathbb{R}^{n\times m}$, and $\mathbb{N}$ be the $n$-dimensional Euclidean space, the set of $n\times m$ real matrices, and the set of nonnegative integers, respectively. 
$I$ represents the identity matrix.
We denote the Euclidean norm of a vector and the spectral norm of a matrix by $\|\cdot\|$.
For a vector $x\in \mathbb{R}^n$, by $x_i$ denote its $i$-th entry,
by $\diag(x)$ the $n$-dimensional diagonal matrix with diagonal entries $x_{1}, x_{2}, \dots, x_{n}$,
and by $\dup(x)$ the vector $[x^{\top}\ x^{\top}]^{\top} \in \RR^{2n}$ obtained by vertically stacking $x$ twice. 
We define the duplication of a set $\mtcx \subseteq \RR^{n}$ as the set of vectors obtained by vertically stacking each vector $x \in \mtcx$ twice, i.e., $\dup(\mtcx) \Let \{\dup(x): x \in \mtcx\} \subseteq \RR^{2n}$. 
For a vector $\xchk \in \RR^{2n}$, we denote the first $n$ entries of $\xchk$ by $\xchk_{1:n}$.
We call a square matrix $A = [a_{ij}]\in \RR^{n\times n}$ stochastic, if $a_{ij} \geq 0$ and $\sum_{k = 1}^{n} a_{ik} = 1$ for $1 \leq i, j \leq n$. 
Denote the set of $n$-dimensional stochastic matrices by $\mtcs_{n}$.
The cardinality of a set $\mtct$ is $|\mtct|$. 
The unit interval $[0,1]$ is denoted by $\mtci$.

\section{Preliminaries}\label{sec:preliminaries}
In this section we introduce the concepts of transition systems and logic languages, and formulate the problem.

\subsection{The Friedkin--Johnsen Model}\label{sec:problem}
We study a parameter reasoning problem in opinion dynamics over social networks, where $n$ agents exchange opinions about a topic. 
The opinion of agent $i$ at time $t$ is denoted by $x_{i}(t)$, and the vector $x(t) \in \mtci^{n}$ collects the opinions of all agents. 
The opinion formation process is modeled by the FJ model~\cite{friedkin1990social}
\begin{equation}\label{eq:model_specific}
        x(t+1) = (I - \Lambda) W x(t) + \Lambda x(0), \text{ for } t \in \mathbb{N},
\end{equation}
where $\Lambda = \diag(\lambda) \in \mtci^{n\times n}$ is a diagonal matrix with the diagonal $\lambda_{i} \in \mtci$ representing the stubbornness of agent~$i$, 
and $W \in \mtcs_{n}$ is the influence weight matrix. 
Note that the above model can be written as follows:
\begin{align}\label{eq:model_specific2}
    \begin{bmatrix}
        x(t+1) \\ z(t+1)
    \end{bmatrix}
    = 
    \begin{bmatrix}
        (I - \Lambda) W & \Lambda \\
        0 & I
    \end{bmatrix}
    \begin{bmatrix}
        x(t) \\ z(t)
    \end{bmatrix},\;
    z(0) = x(0).
\end{align}

In practice, real-valued opinions are often not observable, due to coarse measurements or privacy restrictions.
Instead, quantized observations may be available. 
For example, if an agent has an opinion beyond a threshold $\gamma$, she may report a value of $1$, and otherwise $0$ (a common choice of $\gamma$ is $\gamma = 0.5$).
Overall we thus observe a binary sequence $y(t)$:
\begin{align}\label{eq:binary_observation}
    y(t) = \mtcq(x(t)),\quad ~t \in \NN,
\end{align}
where $\mtcq$ is the quantization function such that $y_{i}(t) = 1$ if $x_{i}(t) \geq \gamma$ and $y_{i}(t) = 0$ otherwise.



\subsection{Transition Systems}
Discrete-time dynamical systems can be formulated as transition systems~\cite{baier2008principles,tabuada2009verification}. In our definition, we omit the set of inputs, since the FJ model does not have inputs.
\begin{defn}[Transition System]
    A transition system is a tuple $\mtct = (\mtcx, X_{0}, \theta, \to, Y, H)$, where
    \begin{itemize}
        \item $\mtcx$ is a set of states,
        \item $X_{0} \subseteq \mtcx$ is a set of initial states,
        \item $\theta \in \Theta$ contains parameters of the system, where $\Theta$ is the parameter set,
        \item $\to \subseteq \mtcx \times \mtcx$ is a transition relation depending on $\theta$,
        \item $Y$ is a set of outputs, and
        \item $H\colon \mtcx \to Y$ is an output map.
    \end{itemize}
    We use $\mtct(X_{0}, \theta)$ to emphasize the dependence on parameters and initial states.
    If $X_{0}$ is a singleton, i.e., $X_{0} = \{x_{\init}\}$, we denote the system by $\mtct(x_{\init}, \theta)$.
\end{defn}

The FJ model with binary observations~\eqref{eq:model_specific2}--\eqref{eq:binary_observation} can then be regarded as a transition system $\mtct(\xchk_{\init}, \theta)$, where $\xchk_{\init} \Let \dup(x(0))$. 
Specifically, the parameters are given by $\theta = (\lambda, W) \in \Theta = \mtci^{n} \times \mtcs_{n}$.
The set of states is $\mtcx \subseteq \mtci^{2n}$
and the set of initial states is $X_{0} = \{\xchk_{\init}\} \subseteq \dup(\mtci^{n})$. 
The transition relation $\to$ is defined such that $(\xchk, \xchk^{+}) \in \to$ if and only if 
\begin{align*}
    \xchk^{+} = 
    \begin{bmatrix}
        (I - \Lambda) W & \Lambda \\
        0 & I
    \end{bmatrix}
    \xchk.
\end{align*}
Finally, the set of outputs is $Y = \{0,1\}^{n}$, and the output map is $H(\xchk) = \mtcq(\xchk_{1:n})$.
The previously defined transition system $\mtct(\xchk_{\init}, \theta)$ with a singleton initial state possesses a unique path (defined in~\cite{baier2008principles}) $\sigma = (\xchk(0), \xchk(1), \dots)$ satisfying that $\xchk(0) = \xchk_{\init}$ and $\xchk(t) = [x(t)^{\tp}~x(0)^{\tp}]^{\tp}$.

\subsection{Approximate Simulation}
We introduce the following normalized Hamming distance as  metric for binary signals $x, y \in \{0,1\}^{n}$,
\begin{align*}
    \|x - y\|_{\txH} = \frac{1}{n} \sum_{i=1}^{n} |x_{i} - y_{i}|.
\end{align*}

For two transition systems defined by the FJ model, we define the following approximate simulation relation, which implies that, for the two systems, both states and binary observations are close.
\begin{defn}[Approximate Simulation Relation]
    Consider two transition systems $\mtct_{a} = (\mtcx_{a}, X_{a0}, \to_{a}, \theta_{a}, \linebreak[4] Y, H)$ and $\mtct_{b} = (\mtcx_{b}, X_{b0}, \to_{b}, \theta_{b}, Y, H)$, where $\mtcx_{a}, \mtcx_{b}$ $ \subseteq \RR^{2n}$, $X_{a0}, X_{b0} \subseteq \dup(\RR^{n})$, and $Y = \{0,1\}^{n}$ is equipped with the normalized Hamming distance. Let $\delta \geq 0$ be a nonnegative real number. A relation $R \subseteq \mtcx_{a} \times \mtcx_{b}$ is an $\delta$-approximate simulation relation from $\mtct_{a}$ to $\mtct_{b}$ if the following conditions are satisfied:
    \begin{itemize}
        \item for every $\xchk^{a0} \in X_{a0}$, there exists $\xchk^{b0} \in X_{b0}$ with $(\xchk^{a0}, \xchk^{b0}) \in R$,
        \item for every $(\xchk^{a}, \xchk^{b}) \in R$, it holds that $\|\xchk^{a}_{1:n} - \xchk^{b}_{1:n}\| \leq \delta \sqrt{n}$ and $\|H(\xchk^{a}) - H(\xchk^{b}) \|_{\txH} \leq \delta$, and 
        \item for every $(\xchk^{a}, \xchk^{b}) \in R$, it holds that $\xchk^{a} \to_{a} (\xchk^{a})^{\prime}$ in $\mtct_{a}$ implies the existence of $\xchk^{b} \to_{b} (\xchk^{b})^{\prime}$ in $\mtct_{b}$ such that $((\xchk^{a})^{\prime}, (\xchk^{b})^{\prime}) \in R$.
    \end{itemize}
    We say that $\mtct_{a}$ is $\delta$-approximately simulated by $\mtct_{b}$ or that $\mtct_{b}$ $\delta$-approximately simulates $\mtct_{a}$, denoted by $\mtct_{a} \preceq_{\delta} \mtct_{b}$.
\end{defn}

\begin{rmk}\label{rmk:apprx_simulation}
The preceding definition introduces stronger conditions than the classic approximate simulation \cite{tabuada2009verification,girard2007approximation}, which only assumes output proximity.
However, in our system, a small output error $\|H(\xchk^{a}) - H(\xchk^{b}) \|_{\txH}$ does not guarantee $x^{a}$ and $x^{b}$ are close.
On the other hand, due to the discontinuity of the quantization function $\mtcq$, two nearby states $x^{a}$ and $x^{b}$ may produce entirely different outputs.
\end{rmk}

\subsection{Specification and Problem Statement}\label{subsec:specification}
We are interested in verifying whether a transition system with certain parameters is consistent with an observed sequence of binary outputs. 
To formally specify temporal output properties of the transition system $\mtct$, we use signal temporal logic (STL)~\cite{raman2015reactive}. 
An STL formula $\varphi$ consists of a finite set of predicates $\mu$, a set of logical operators, and a set of temporal operators.
The syntax of STL can be described as:
\begin{align*}
\varphi ::=  true \; |\; \mu \; |\; \neg \varphi\;  |\; \varphi_{1} \wedge \varphi_{2} \; |\; \varphi_{1} \mathbf{U}_{[a, b]} \varphi_{2}, 
\end{align*}
where $\neg$ is the negation operator, $\wedge$ is the conjunction operator, and $\mu$ is a predicate defined through a predicate function $g_{\mu}\colon Y \to \RR$ as
\begin{align*}
    \mu \Let \begin{cases}
        true, & \tx{ if } g_{\mu}(y) \geq 0, \\
        false, & \tx{ if } g_{\mu}(y) < 0,
    \end{cases} \quad \tx{ for } y \in Y.
\end{align*}
The disjunction operator can be then defined as $\varphi_{1} \vee \varphi_{2} = \neg (\neg \varphi_{1} \wedge \neg \varphi_{2})$.
For a trajectory $\bfy = \{y(0),~ y(1),~ \dots\}$ and time $t \geq 0$, we say $\bfy$ satisfies $\mu$ at $t$, if $g_{\mu}(y(t)) \geq 0$, i.e., $(\bfy, t) \models \mu \tx{ iff } g_{\mu}(y(t)) \geq 0$.
Inductively, 
\begin{align*}
    &(\bfy, t) \models \neg \varphi &\tx{ iff } &\neg ((\bfy, t) \models \varphi), \\
    &(\bfy, t) \models \varphi_{1} \wedge \varphi_{2} &\tx{ iff } &(\bfy, t) \models \varphi_{1} \wedge (\bfy, t) \models \varphi_{2}, \\
    &(\bfy, t) \models \varphi_{1} \mathbf{U}_{[a, b]} \varphi_{2} &\tx{ iff } &\exists t^{\prime} \in [t + a, t + b] \tx{ s.t. } (\bfy, t) \models \varphi_{2} \\
    & & &\wedge \forall t^{\prime\prime} \in [t, t^{\prime}], (\bfy, t) \models \varphi_{1}. 
\end{align*}
The property that the output sequence is compatible with the observed ones $\{\tilde{y}(t), 0 \leq t \leq T\}$ is encoded by a sequence of predicates $\Phi^{\kappa} = \{\varphi^{\kappa}(t), 0 \leq t \leq T\}$ with $g_{\varphi^{\kappa}(t)}(y) \Let \kappa - \|y - \tilde{y}(t)\|_{\txH}$, where $\kappa \geq 0$. 
For the output sequence $\bfy  = \{y(t), 0 \leq t \leq T\}$ of a transition system $\mtct(\xchk, \theta)$, define $\bfy \models \Phi^{\kappa}$ iff $(\bfy, t) \models \varphi^{\kappa}(t)$ for all $0 \leq t \leq T$.
Since $\mtct(\xchk, \theta)$ admits only a unique path, we say $\mtct(\xchk, \theta) \models \Phi^{\kappa}$ if $\bfy \models \Phi^{\kappa}$, where $\bfy$ corresponds to its unique output sequence.
The following problem is studied in this paper.

{\bf Problem.} Given a sequence of binary observations from the FJ model $\{\tilde{y}(0), \dots, \tilde{y}(T)\}$, a set of initial states $X_{0}$, a set of parameters $\Theta_{0}$, an estimate of the influence matrix $W$, and a nonnegative real number $\kappa \geq 0$, verify wheher there exists $(\xchk, \theta) \in X_{0} \times \Theta_{0}$ such that 
\begin{align*}
    \mtct(\xchk, \theta) \models \Phi^{\kappa}.
\end{align*}
The satisfaction of $\Phi^{\kappa}$ by $\mtct(\xchk, \theta)$ means that the binary output $y(t)$ is consistent with the observed $\tilde{y}(t)$ up to an error $\kappa$ for all $t \leq T$, where $\kappa$ represents the tolerance level. When $\kappa = 0$, an exact match is required, whereas a small $\kappa > 0$ provides robustness in the verification.
We construct abstractions that approximately simulate the original system (Theorem~\ref{thm:apprx_simul}), and develop a verification framework through abstraction (Proposition~\ref{prop:verification}).

\section{Main Results}\label{sec:result}
In this section, we provide conditions under which the original transition system $\mtct(\xchk_{\init}, \theta) = (\mtcx, X_{0}, \theta, \to, Y, H)$ \linebreak[4] can be $\delta$-approximately simulated by an abstraction. 
This makes it possible to transfer the previously defined verification problem for the original system to its abstraction.

We construct the following abstraction to approximate the original transition system $\mtct(\xchk_{\init}, \theta)$.
The abstraction $\mtctab$ is a tuple $\mtctab = (\mtcXab, \Xab_{0}, \thetaab, \toab, \Yab, \Hab)$, where
\begin{itemize}
    \item $\mtcXab \subseteq \mtci^{2n}$ is the set of states,
    \item $\Xab_{0}$ is the set of initial states, which will be defined in Assumption~\ref{asmp:apprx},
    \item $\thetaab = (\lambdaab, \Wab) \in \Thetaab$ is the parameters of the system, with $\Thetaab$ the set of parameters, which will be defined in Assumption~\ref{asmp:apprx},
    \item $\toab \subseteq \Xab \times \Xab$ is the transition relation, such that $(\xchkab, (\xchkab)^{+}) \in \toab$ if and only if
    \begin{align*}
        (\xchkab)^{+} = 
        \begin{bmatrix}
        (I - \Lambdaab) \Wab & \Lambdaab \\
        0 & I
        \end{bmatrix} \xchkab,
    \end{align*}
    where $\Lambdaab = \diag(\lambdaab)$,
    \item $\Yab = Y$ is the output set, and
    \item $\Hab(\xchkab) = \mtcq(\xab_{1:n})$ is the output map.
\end{itemize}

In the following assumptions, we introduce 
sets of parameters and initial states, which are used to construct abstractions that approximately simulate behaviors of the original system.
\begin{asmp}\label{asmp:apprx} For the abstraction $\mtctab$, let $\Thetaab = \Thetaab_{\lambda} \times \Thetaab_{w}$ be the set of parameters, and $d_{x}$ and $d_{\lambda}$ be positive integers. Suppose that the following conditions hold.
    \begin{enumerate}
        \item The set of initial states $\Xab_{0}$ satisfies that $\Xab_{0} \subseteq \mtcx_{0}^{\ab}(d_{x})$, where $\mtcx_{0}^{\ab} \Let \dup ( \{\frac{1}{2\dx}, \frac{3}{2\dx}, \dots, \frac{2\dx - 1}{2\dx} \}^{n} )$.
        \item The stubbornness of each agent in the abstraction $\mtctab$ takes value in the set $\mtcl(\dlambda) \Let \{0, \frac{1}{\dlambda}, \dots, \frac{\dlambda - 1}{\dlambda},  1 \}$, i.e., $\Thetaab_{\lambda} = \mtcl(\dlambda)^{n}$, and $\lambdaab \in \Thetaab_{\lambda}$. 
        \item The set of weight matrix $\Thetaab_{w} = \{\Wab\}$ with $\Wab \in \mtcs_{n}$ such that 
            \begin{align}\label{eq:asmp_W}
                \|\Wab - W\| \leq \eps_{w},
            \end{align}
        where $W$ is the weight matrix of the original system.
    \end{enumerate}
\end{asmp}

\begin{rmk}
    We discretize the stubbornness parameters and initial states. 
    We also assume that the network structure can be observed, subject to certain noise. 
    When the network is generated from a random graph model, the  normalized adjacency matrices are known to concentrate around averaged versions~\cite{schaub2020blind}, which is a property described by~\eqref{eq:asmp_W}.
\end{rmk}

For the transition system $\mtct(\xchk_{\init}, \theta)$, recall that $\xchk_{\init} = \dup(x(0))$. 
The next lemma indicates the existence of an abstraction, whose parameters and initial states are close to those of the original system, such that the states of the two systems remain close. 

\begin{lem}\label{lem:apprx}
    Suppose that Assumption~\ref{asmp:apprx} holds, and the following conditions hold for $\thetaab \in \Thetaab$ and $\xchkab_{\init} \in \mtcx^{\ab}_{0}$
    \begin{align}\label{eq:asmp_x0}
        \|x_{\init} - \xab_{\init}\| \leq \frac{\sqrt{n}}{2\dx}, \\ \label{eq:asmp_lambda}
        \|\Lambda - \Lambdaab\| \leq \frac{1}{2\dlambda},
    \end{align}
    where $x_{\init} = [\xchk_{\init}]_{1:n}$, $\xab_{\init} = [\xchkab_{\init}]_{1:n}$, $\Lambda = \diag(\lambda)$ and $\Lambdaab = \diag(\lambdaab)$.
    Then the paths $\{\xchk(t)\}$ and $\{\xchkab(t)\}$ of $\mtct(\xchk_{\init}, \theta)$ and $\mtctab( \xchkab_{\init}, \thetaab)$ satisfy that, for $t \geq 0$,
    \begin{align}\notag
        &\|x(t+1) - \xab(t+1)\| \\\label{eq:one_step_bound}
        &\leq \|(I - \Lambda) W\| \|x(t) - \xab(t)\| + \sqrt{n} \eps_{x},
    \end{align}
    where $x(t) = \xchk_{[1:n]}(t)$, $\xab(t) = \xchkab_{[1:n]}(t)$, and
    \begin{align}\label{eq:eps_x}
        \eps_{x} = \frac{\|W\| + 1}{2\dlambda} + \frac{1}{2\dx} + \eps_{w}. 
    \end{align}
    If it further holds that $\|(I - \Lambda) W\| < 1$, then 
    \begin{align*}
        \sup_{t\geq 0} \|x(t) - \xab(t)\| \leq \frac{\eps_{x} \sqrt{n}}{1 - \|(I - \Lambda) W\|}.
    \end{align*}
\end{lem}
\begin{proof}
    It follows from~\eqref{eq:asmp_x0} and~\eqref{eq:asmp_lambda} that
    \begin{align*}
        &\|\Lambda - \Lambdaab\| 
        \le \frac{1}{2\dlambda},\\
        &\|(I - \Lambda) W - (I - \Lambdaab) \Wab\| \\
        &\leq \|(I - \Lambdaab) (W - \Wab)\| + \|(\Lambda - \Lambdaab) W\| \\
        &\leq \eps_{w} + \frac{\|W\|}{2\dlambda} .
    \end{align*}
    Note that $x_{\init} = x(0)$, $\xab_{\init} = \xab(0)$, and
    \begin{align*}
        &\|x(t+1) - \xab(t+1)\| \\
        &= \| (I - \Lambda) W x(t) + \Lambda x(0) \\
        &\quad\ - (I - \Lambdaab) \Wab \xab(t) - \Lambdaab \xab(0) \| \\\notag
        &= \| (I - \Lambda) W (x(t) - \xab(t)) \\ 
        &\quad\ + ((I - \Lambda) W - (I - \Lambdaab) \Wab) \xab(t)  \\\notag
        &\quad\ + (\Lambda - \Lambdaab) x(0)  + \Lambdaab (x(0) - \xab(0))\| \\
        &\leq \|(I - \Lambda) W\| \|x(t) - \xab(t)\|+ \bigg(\frac{\|W\|}{2\dlambda} + \eps_{w} \bigg) \|\xab(t)\| \\
        &\quad\ +  \frac{1}{2\dlambda} \|x(0)\| + \frac{\sqrt{n}}{2\dx}  \\
        &\leq \|(I - \Lambda) W\| \|x(t) - \xab(t)\| \\
        &\quad\ + \sqrt{n} \bigg(\frac{\|W\| + 1}{2\dlambda} + \frac{1}{2\dx} + \eps_{w} \bigg).
    \end{align*}
    The conclusion follows from induction.
\end{proof}
\begin{rmk}\label{rmk:lem:approx}
    The previous lemma indicates that the state approximation of the original system by the abstraction is bounded by an error depending on $\eps_{x}$, when $(I - \Lambda) W$ is a contraction. 
    This value vanishes as $\dlambda$ and $\dx$ grows and $\eps_{w}$ decreases.
    When the error $\eps_{x}$ is small enough, at each time $t$, most entries of $x(t)$ and $\xab(t)$ are close to each other.
\end{rmk}

As noted in Remark~\ref{rmk:apprx_simulation}, the discontinuity in the quantization function complicates the characterization of the relation between opinions and their binary outputs.
To obtain the approximate simulation relation, we introduce the following index set of a vector $x$ whose entries are close to the quantization threshold $\gamma$.
For a vector $x \in \RR^{n}$, we define $\mtcv(x, \eta) \Let \{i \in \mtcv: |x_{i} - \gamma| \leq \eta\}$ as the set of indices $i$ such that the absolute difference between $x_{i}$ and $\gamma$ is less than $\eta$.
Since the $\mtct(\xchk_{\init},\theta)$ has only a unique path, hereafter we restrict its state space to the set of states visited along this path, denoted by $\mtcx_{\tx{res}}$.
The following technical assumption ensures that only a few agents have opinions close to the threshold $\gamma$ for the given time period of observation.
Relaxing this assumption is left to future work.
\begin{asmp}\label{asmp:prop}
    Consider the transition system $\mtct(\xchk_{\init},\theta)$ with an initial state $\xchk_{\init} \in \dup(\mtci^{n})$ and parameters $\theta \in \mtci^{n} \times \mtcs_{n}$.
    Assume that there exists a constant $\delta > 0$ such that $|\mtcv(\xchk_{1:n}, \sqrt{2 \delta})| \le \delta n / 2$ for all $\xchk \in \mtcx_{\tx{res}}$.
\end{asmp}

The following theorem shows that the original system can be approximately simulated by the abstraction in Lemma~\ref{lem:apprx}.

\begin{thm}\label{thm:apprx_simul}
    Consider the transition system $\mtct(\xchk_{\init}, \theta)$ and an abstraction $\mtctab(\xchkab_{\init}, \thetaab)$ defined in Lemma~\ref{lem:apprx}.
    Suppose that Assumptions~\ref{asmp:apprx} and~\ref{asmp:prop} hold. 
    If the conditions below hold
    \begin{align}\label{eq:asmp_dx}
        &\dx \geq 1 / 2 \delta, \\ \label{eq:asmp_lambda_w}
        &\|(I - \Lambda) W\| < 1, \\ \label{eq:apprx_condition}
        &\eps_{x} \leq (1 - \|(I - \Lambda) W\|) \delta,
    \end{align}
    then $\mtct(\xchk_{\init}, \theta) \preceq_{\delta} \mtctab(\xchkab_{\init}, \thetaab)$, where $\eps_{x}$ is given in~\eqref{eq:eps_x}.
\end{thm}
\begin{proof}
    For $(\xchk, \xchkab) \in R$ with $\|x - \xab\| \leq \delta \sqrt{n}$, $\xchk \to \xchk^{+}$, and $\xchkab \toab (\xchkab)^{+}$, it follows from~\eqref{eq:one_step_bound} that
    \begin{align*}
        &\|x^{+} - (\xab)^{+}\| \\
        &\leq \|(I - \Lambda) W\| \delta \sqrt{n} + \sqrt{n} \eps_{x}
        \leq \delta \sqrt{n}.
    \end{align*}
    This bound ensures that
    \begin{align*}
        \|x^{+}|_{(\mtcv(\sqrt{2\delta}))^{\tx{c}}} - (\xab)^{+}|_{(\mtcv(\sqrt{2\delta}))^{\tx{c}}}\| \leq \|x^{+} - (\xab)^{+}\| \leq \delta \sqrt{n},
    \end{align*}
    where for a vector $x$ and an index set $\mtcs$, $x|_{\mtcs}$ represents the restriction of the vector to the entries whose indices are in $\mtcs$, and for a set $\mtcs \subseteq \mtcv$ we denote $\mtcs^{\tx{c}} = \mtcv \setminus \mtcs$.
    This upper bound implies that, except for $\delta n/2$ agents in $(\mtcv(\sqrt{2\delta}))^{\tx{c}}$, every agent $i$ satisfies that $|x^{+}_{i} - (\xab)^{+}_{i}| < \sqrt{2\delta}$.
    From Assumption~\ref{asmp:prop}, each agent $i \in (\mtcv(\sqrt{2\delta}))^{\tx{c}}$ (which has cardinality at least $(1 - \delta/2) n$) satisfy that $(x^{+}_{i} - \gamma)((\xab)^{+}_{i} - \gamma) > 0$. 
    Hence $|y^{+}_{i} - (\yab)^{+}_{i}| = 0$ for at least $(1 - \delta) n$ agents, meaning $\|y^{+} - (\yab)^{+}\|_{\txH} \leq \delta$.    

    For $\xchk_{\init} \in X_{0}$, it follows from Lemma~\ref{lem:apprx} that $\|\xab_{\init} - x_{\init}\| \leq \sqrt{n}/(2\dx) \leq \delta \sqrt{n}$. A similar argument implies that $\|H(\xchk_{\init}) - H(\xchkab_{\init})\|_{\txH} \leq \delta$, and hence $(\xchk_{\init}, \xchkab_{\init}) \in R$.
\end{proof}

Note that~\eqref{eq:asmp_dx} provides an upper bound for $\delta$ and~\eqref{eq:apprx_condition} provides a lower bound.
Finer abstractions and a smaller $\|(I - \Lambda) W\|$ increase the admissible range of $\delta$.

The assumptions in Theorem~\ref{thm:apprx_simul} restrict the set of systems we can reason about.
Define 
\[\Pi_{1} \Let \{(\xchk_{\init}, \theta) \in \dup(\mtci^{n}) \times (\mtci^{n} \times \mtcs_{n}): \|(I - \Lambda) W \| < 1\}\] 
as the set of parameters satisfying condition~\eqref{eq:asmp_lambda_w}.
Define
\begin{align*}
    \Pi_{2}(\delta) &\Let \{(\xchk_{\init}, \theta) \in \dup(\mtci^{n}) \times (\mtci^{n} \times \mtcs_{n}): \\
    &\qquad\ \mtct(\xchk_{\init}, \theta) \tx{ satisfies Assumption~\ref{asmp:prop}}\}.
\end{align*}
Given $\delta > 0$, each abstract configuration of an initial state and parameters $(\xchkab_{\init}, \thetaab)$ can represent a set of original configurations:
\begin{align*}
    &\Pi_{3}(\xchkab_{\init}, \thetaab, \delta) \Let \{(\xchk_{\init}, \theta) \in \dup(\mtci^{n}) \times (\mtci^{n} \times \mtcs_{n}):  \\
    & \quad\ \exists 1 / (2 \delta) \leq \dx \in \NN, 1 \leq \dlambda \in \NN, \eps_{w} \geq 0 \tx{ satisfying}~\eqref{eq:apprx_condition} \\\notag
    & \quad\ \tx{ such that conditions }  \eqref{eq:asmp_W}\tx{--}\eqref{eq:asmp_lambda} \tx{ hold} \}.
\end{align*}
Denote $\Pi(\xchkab_{\init}, \thetaab, \delta) = \Pi_{1} \cap \Pi_{2}(\delta) \cap \Pi_{3}(\xchkab_{\init}, \thetaab, \delta)$, and we have the following consequence of the $\delta$-approximate simulation relation given in Theorem~\ref{thm:apprx_simul}.
\begin{prop}\label{prop:trans}
    Suppose that Assumption~\ref{asmp:apprx} holds. 
    Then for $(\xchkab_{\init}, \thetaab) \in \mtcx^{\ab}_{0} \times \Thetaab$, any $\delta > 0$ such that $\Pi(\xchkab_{\init}, \thetaab, \delta)$ is nonempty, and any $\kappa \geq 0$, $\mtctab(\xchkab_{\init}, \thetaab) \not\models \Phi^{\kappa+\delta}$ implies that $\mtct( \xchk_{\init}, \theta) \not\models \Phi^{\kappa}$ for all $(\xchk_{\init}, \theta) \in \Pi(\xchkab_{\init}, \thetaab, \delta)$.
\end{prop}
\begin{proof}
    Suppose that there exists $(\xchk_{\init}, \theta) \in \Pi(\xchkab_{\init}, \thetaab, \delta)$ such that $\mtct( \xchk_{\init}, \theta) \models \Phi^{\kappa}$. The satisfaction relation implies that $\|y(t) - \tilde{y}(t)\|_{\txH} \leq \kappa$, $\forall 0 \leq t \leq T$. 
    The definition of $\Pi(\xchkab_{\init}, \thetaab, \delta)$ and Assumption~\ref{asmp:apprx} indicate that $\mtct(\xchk_{\init}, \theta) \preceq_{\delta} \mtctab(\xchkab_{\init}, \thetaab)$, which means that $\|y(t) - \yab(t)\|_{\txH} \leq \delta$, $\forall 0 \leq t \leq T$.
    Hence by the triangle inequality, $\|\yab(t) - \tilde{y}(t)\|_{\txH} \leq \delta+\kappa$, $\forall 0 \leq t \leq T$, leading to a contradiction.
\end{proof}

The previous result makes it possible to test whether the output sequence of the transition system with a given set of parameters is consistent with observations. 
In particular, to check whether the parameters and initial states in a given set are consistent with observations, we can verify $\Phi^{\kappa+\delta}$ for abstractions derived from an abstract configuration set that covers the original configuration set.
\begin{rmk}
    The value of $\delta$ depends on both Assumption~\ref{asmp:prop} and the condition~\eqref{eq:apprx_condition}.
    As noted in Remark~\ref{rmk:lem:approx}, $\eps_{x}$ can be made small when $d_{\lambda}$ and $d_{x}$ are sufficiently large and $\eps_{w}$ is small.
    This ensures the existence of $\delta$ such that~\eqref{eq:apprx_condition} holds, and thus that the set $\Pi_{3}$ is nonempty.
    The existence of $\delta$ and hence the nonemptiness of $\Pi$ rely on the initial state and the parameters of the FJ model. 
    Deriving explicit conditions for $\delta$ remains future work.
\end{rmk}

Proposition~\ref{prop:trans} implies the first part of the following result, whereas the second part follows from an argument similar to the proof of Theorem~\ref{thm:apprx_simul}.
\begin{prop}\label{prop:verification}
    Let $W \in \mtcs_{n}$ be a fixed influence matrix $W$, and let $\Wab$ be an estimate of $W$ such that $\|W - \Wab\| \leq \eps_{w}$ with $\eps_{w} \geq 0$.
    Consider $\kappa \geq \delta > 0$ and a set of state-parameter configurations $\Pi_{*} \subseteq (\dup(\mtci^{n}) \times (\mtci^{n} \times \mtcs_{n})) \cap \Pi_{1} \cap \Pi_{2}(\delta)$.
    Assume $1 / (2 \delta) \leq \dx \in \NN$ and $1 \leq \dlambda \in \NN$ satisfy~\eqref{eq:apprx_condition}, and $0 \leq \eps_{w} < (1 - \inf_{\Pi_{*}} \|(I - \Lambda) W\|) \delta$.
    Then 
    \begin{enumerate}
        \item $\mtctab(\xchkab_{\init}, \thetaab) \not\models \Phi^{\kappa+\delta}$ for all $(\xchkab_{\init}, \thetaab) \in [\Pi_{*}]_{(\dx, \dlambda)}$, 
        implies $\mtct(\xchk_{\init}, \theta) \not\models \Phi^{\kappa}$, $\forall (\xchk_{\init}, \theta) \in \Pi_{*}$,
        where 
        \begin{align*}
            &[\Pi_{*}]_{(\dx, \dlambda)} \Let \{(\xchkab_{\init}, \thetaab) \in \mtcx^{\ab}_{0} \times (\mtcl(\dlambda)^{n} \times \{\Wab\})\colon \\
            &\ (\xchkab_{\init}, \Lambdaab) \tx{ satisfies}~\eqref{eq:asmp_x0}\tx{--}\eqref{eq:asmp_lambda} \tx{ for some } (\xchk_{\init}, \theta) \in \Pi_{*} \}.
        \end{align*}
        \item If there exists $(\xchkab_{\init}, \thetaab) \in [\Pi_{*}]_{(\dx, \dlambda)}$ satisfying Assumption~\ref{asmp:prop} and $\mtct(\xchkab_{\init}, \thetaab) \models \Phi^{\kappa - \delta}$, then there exists $(\xchk_{\init}, \theta) \in \Pi_{*}$ such that $\mtct(\xchk_{\init}, \theta) \models \Phi^{\kappa}$.
    \end{enumerate}
\end{prop}
From Proposition \ref{prop:verification}, it follows that the abstraction yields a necessary condition for consistency with observations.
Violation of the specification by the abstraction guarantees violation by the original system, but satisfaction by the abstraction indicates potential consistency.
Therefore the verification problem introduced in Section~\ref{subsec:specification} can be addressed by verifying whether the output sequence of the abstraction is consistent with the binary observations $\{\tilde{y}(0), \dots, \tilde{y}(T)\}$.

\section{Numerical Experiments}\label{sec:simul}
In this section we present numerical experiments to illustrate the obtained results. 
To address the verification problem, we use the SMT solver Z3 (version 4.15.3), as all the constraints can be expressed in algebraic form.

Consider the FJ model~\eqref{eq:model_specific} with binary observations~\eqref{eq:binary_observation}, where the number of agents is $n = 40$.
The initial opinion $x_{i}(0)$ and the stubbornness $\lambda_{i}$ of each agent $i$ are independently sampled from the uniform distribution on $\mtci$. 
The adjacency matrix of the network is generated from a stochastic block model, where the agents are divided into two equal-sized communities. 
Two agents are connected with probability $0.3$ if they are in the same community, and with probability $0.1$ otherwise.
The influence matrix $W$ is obtained by row-normalizing the adjacency matrix.
We construct abstractions under Assumption~\ref{asmp:apprx}, where the number of abstract initial states is $\dx$ and the number of abstract stubbornness is $\dlambda$.
As shown in Fig.~\ref{fig:trajectory}, when $\dx$ and $\dlambda$ increases, the trajectories of the abstractions become closer to the original system, where the same influence matrix is used in the abstractions.
We further calculate the approximation error between the states, $\max_{1 \leq t \leq T} \|x(t) - \xab(t)\|$, and between the outputs, $\max_{1 \leq t \leq T} \|y(t) - \yab(t)\|_{\txH}$, where the trajectory length $T = 9$.
Fig.~\ref{fig:error} shows the result where the abstractions have the same influence matrix as the original system, whereas Fig.~\ref{fig:error_EW} shows the case where the abstract influence matrix is obtained by row-normalizing the expected adjacency matrix.
In both cases, the error decreases as $\dx$ and $\dlambda$ grows.
The errors are also comparable, indicating that the adjacency matrix may be replaced by its expected version.

The previous example intuitively illustrates that abstractions can approximately simulate the original system when $\dx$ and $\dlambda$ is large and $\eps_{w}$ is small, which validates Proposition~\ref{prop:trans}.
To further illustrate the latter result, we generate a sequence of binary outputs $\{\tilde{y}(t), 0 \leq t \leq T\}$ with $T = 9$ from a system $\mtct(\xchk, (\lambda, W))$ with initial opinion $\xchk$, stubbornness vector $\lambda$, and influence matrix $W$.
Then we sample an abstract initial opinion $\xchkab$ and an abstract stubbornness parameter $\Lambdaab$, which satisfy Assumption~\ref{asmp:apprx}.
By setting $\Wab = W$, we define an abstraction $\mtctab = (\xchkab, (\lambdaab, \Wab))$.
This abstraction generates an output sequence $\{y(t), 0 \leq t \leq T\}$, which has a large error $\|y(t) - \tilde{y}(t)\|_{\txH}$, as shown in Fig.~\ref{fig:error_t_alpha}.
To see the change in observation error, we consider systems $\mtct(\xchk(\alpha), (\lambdaab(\alpha), W))$, where $\xchk(\alpha) = (1 - \alpha) \xchkab + \alpha \xchk$ and $\lambdaab(\alpha) = (1 - \alpha) \lambdaab + \alpha \lambda$ are linear combinations of the abstract and original parameters with $\alpha \in [0,1]$.
Fig.~\ref{fig:error_t_alpha} shows that systems with small $\alpha$ generate outputs with errors similar to those of the abstract configuration $(\xchkab, \lambdaab)$, indicating that the abstraction is a representative of proximal systems.

\begin{figure}[tbp]
    \centering
    \subfigure[\label{fig:trajectory} Trajectories of three agents in the FJ model and two abstractions. The solid lines represent the original FJ model, the dashed lines represent the abstraction with $\dx = 2$ and $\dlambda = 3$, and the dotted lines represent the abstraction with $\dx = \dlambda = 6$. ]{~~~~~~~~~~~~~~~
        \includegraphics[width=0.2\textwidth]{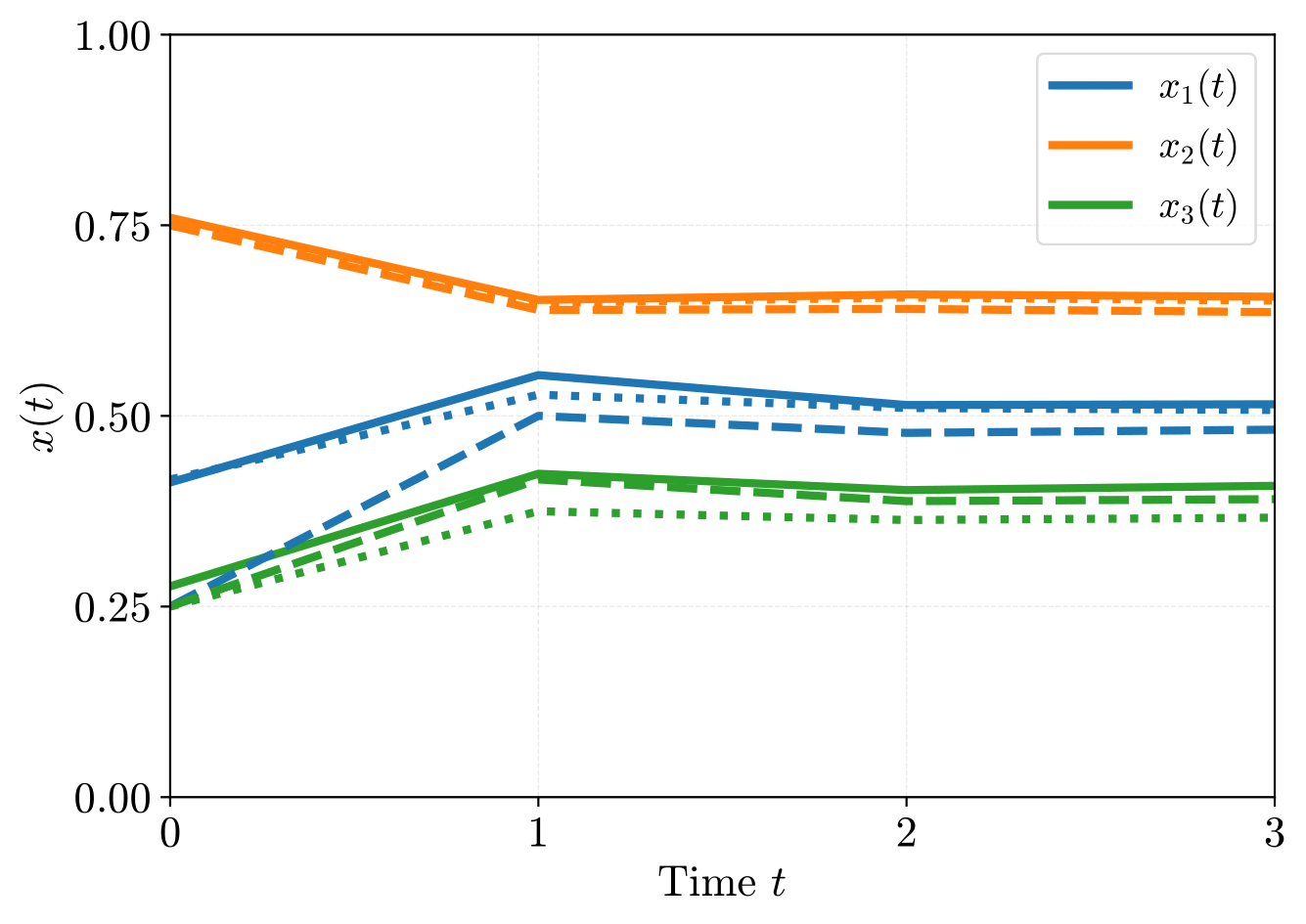} ~~~~~~~~~~~~~~~
    }
    \subfigure[\label{fig:error}The errors between the states (left) and the outputs (right) of the FJ model and its abstractions, where the abstractions have the same influence weight matrix as the original model. ]{ 
        \includegraphics[width=0.19\textwidth]{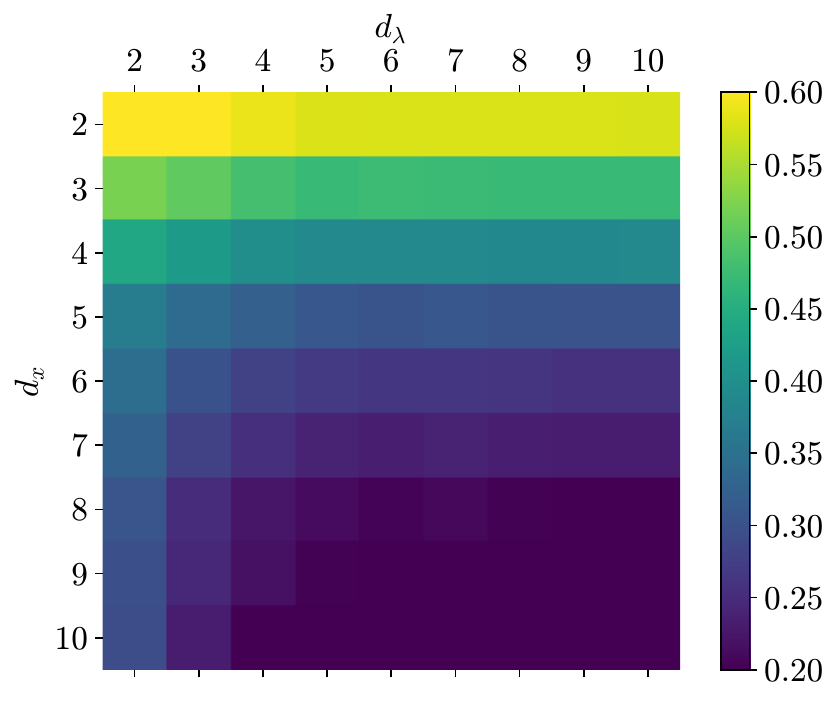} ~~~
        \includegraphics[width=0.19\textwidth]{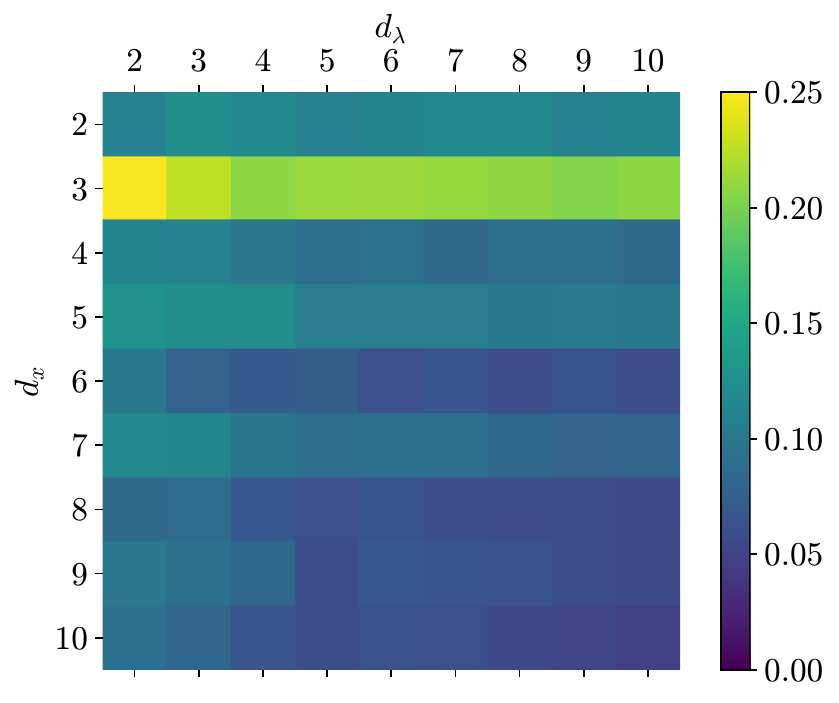}  
    }
    \subfigure[\label{fig:error_EW}The errors between the states (left) and the outputs (right) of the FJ model and its abstractions, where the influence matrix of the abstractions is obtained by row-normalizing the expected adjacency matrix. ]{ ~~~
        \includegraphics[width=0.19\textwidth]{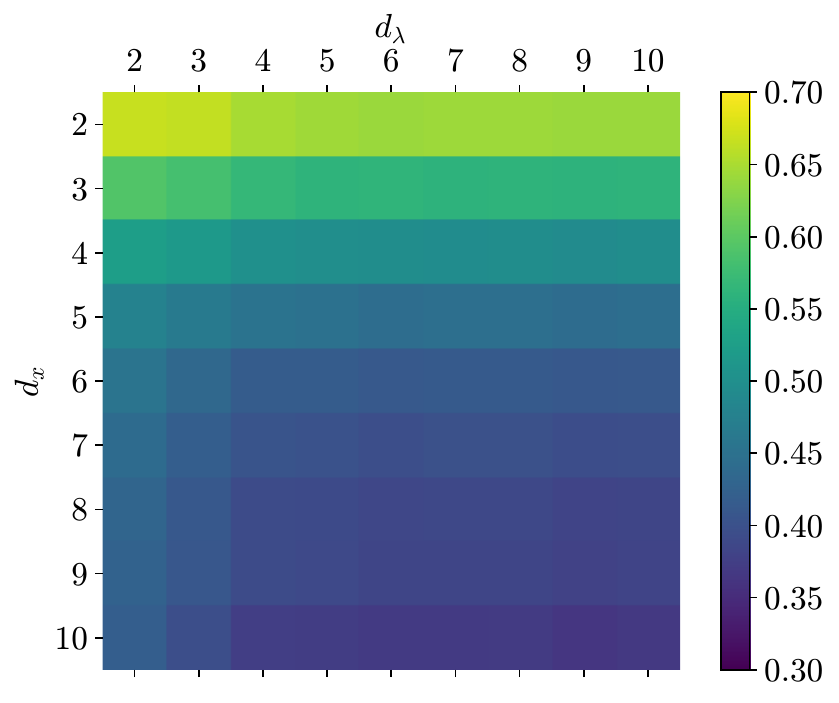} ~~~
        \includegraphics[width=0.19\textwidth]{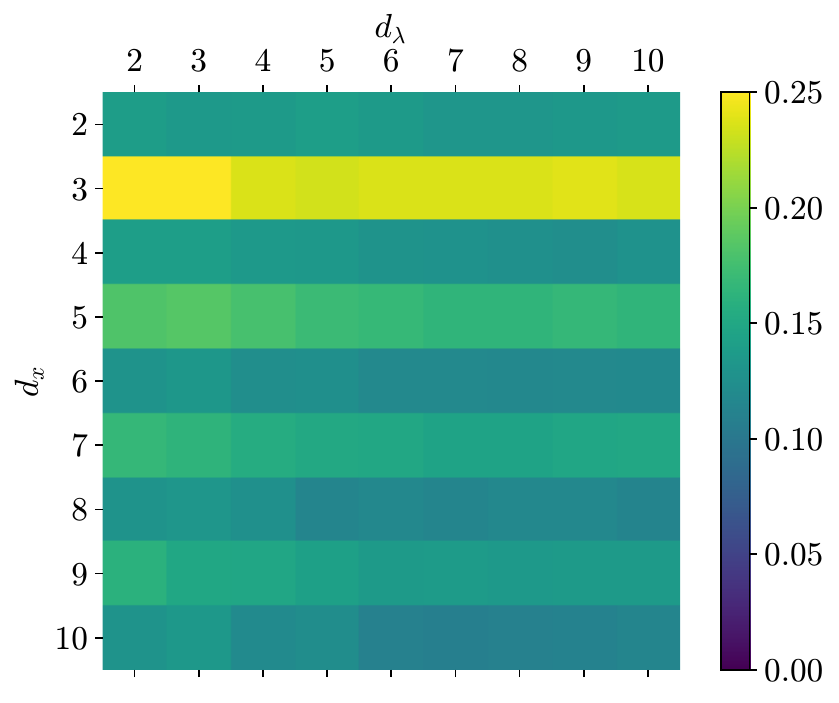}  ~~~
    }
    \caption{\label{fig:sample}Comparison between the FJ model and its abstractions.}
    \vspace{-0.75cm}
\end{figure}

\begin{figure}[tbp]
    \centering
    \includegraphics[width=0.26\textwidth]{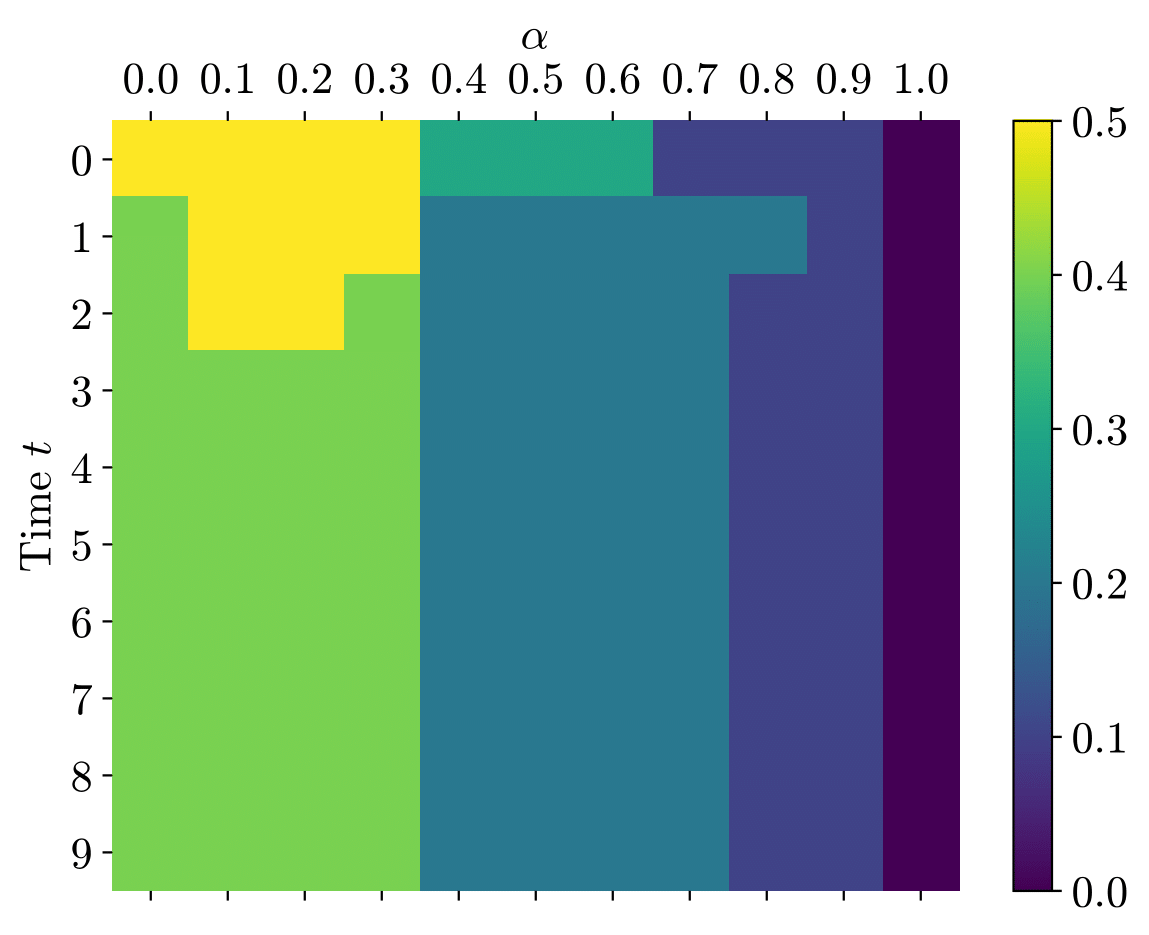} 
    \caption{\label{fig:error_t_alpha} Output distance at each step $t$ between the systems $\mtct(\xchk(\alpha), (\Lambdaab(\alpha), W))$ and $\mtct(\xchk, (\Lambda, W))$, where $\xchk(\alpha) = (1 - \alpha) \xchkab + \alpha \xchk$ and $\Lambdaab(\alpha) = (1 - \alpha) \Lambdaab + \alpha \Lambda$. Each column corresponds to one abstraction.}
\end{figure}

Note that the binary observations provide limited information, and multiple parameter configurations can define FJ models producing identical output sequences.
Consider an example system with $n = 10$, generating a binary output sequence of length $T = 9$.
The initial opinions are sampled uniformly from the discrete set $\mtcx_{0}^{\ab}(d_{x})$ with $\dx = 2$, and the stubbornness from the set $\mtcl(\dlambda)$ with $\dlambda = 3$.
The model checker can return multiple valid parameter configurations as solutions, if it is only asked to verify whether there exists a system that generates the binary observations.
To reduce the number of solutions, one can introduce a constraint $\|\Lambda - \hat{\Lambda}\| \leq \eps_{\Lambda}$, where $\Lambda = \diag(\lambda)$ represents the true parameters, and $\hat{\Lambda} = \diag(\hat{\lambda})$ is a solution.
Fig.~\ref{fig:solutions} shows that the number of solutions increases as $\eps_{\Lambda}$ becomes larger.

Finally, we show how incorporating structural information can reduce the complexity of the model checking.
An FJ model with $n = 200$ is used to generate an output sequence with $T = 3$. 
The initial opinions are sampled uniformly from the discrete set $\mtcx_{0}^{\ab}(d_{x})$ with $\dx = 2$. 
For stubbornness, we assume that $20\%$ of the agents are totally stubborn (i.e., $\lambda_{i} = 1$), and for the remaining agents, we sample $\lambda_{i}^{*}$ from the set $\mtcl(\dlambda)$ with $\dlambda = 3$. 
The network has a block structure with two equal-sized communities, and the weighted adjacency matrix $A = [a_{ij}]$ is given by $a_{ij} = 5$ if $i$ and $j$ are in the same community, and $a_{ij} = 3$ otherwise. 
The influence matrix $W$ is then obtained by row-normalizing $A$.
Such simplification is valid, as demonstrated by the first experiment.
The problem is to check whether there exists a configuration consistent with the outputs such that the stubborn parameter $\lambda_{i}$ of each agent $i$ is within the interval $[\lambda_{i}^{*} - \eps, \lambda_{i}^{*} + \eps] \cap \mtci$ with $\eps = 0.5$, i.e., whether there exists $(\xchk, \theta) \in \mtcx^{\ab}_{0} \times ((\prod_{i} \mtcl_{i}) \times \{W\})$ such that $\mtct(\xchk, \theta) \models \Phi^{\kappa}$, where $\mtcl_{i} = \mtcl(\dlambda) \cap [\lambda_{i}^{*} - \eps, \lambda_{i}^{*} + \eps]$.
The solver finds a solution in all sampled cases, with execution time given in Fig.~\ref{fig:runtime}.
Now assume that the stubborn agents and their initial opinions are known, representing a scenario where opinion leaders are clearly identified.
Since the community structure is known and agents can take only two values ($\dx = 2$), we divide the stubborn agents into four groups and simplify the model, by assigning a common representative state to each group of stubborn agents.
The system dimension is reduced by using the group cardinality, and the execution time decreases in most cases (see Fig.~\ref{fig:runtime}).

\begin{figure}[tbp]
    \centering
    \includegraphics[width=0.22\textwidth]{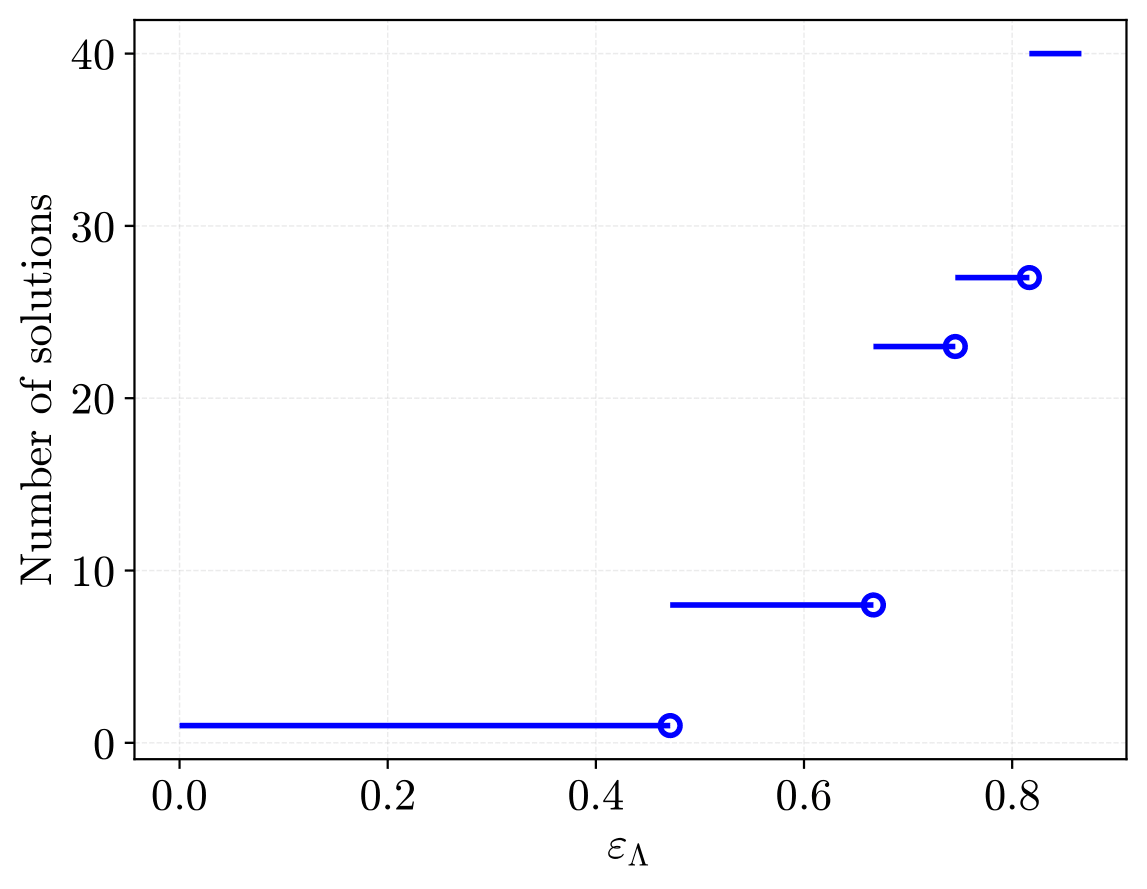} 
    \caption{\label{fig:solutions} The number of solutions to the verification problem under the constraint $\|\Lambda - \hat{\Lambda}\| \leq \eps_{\Lambda}$.}
\end{figure}

\begin{figure}[tbp]
    \centering
    \includegraphics[width=0.22\textwidth]{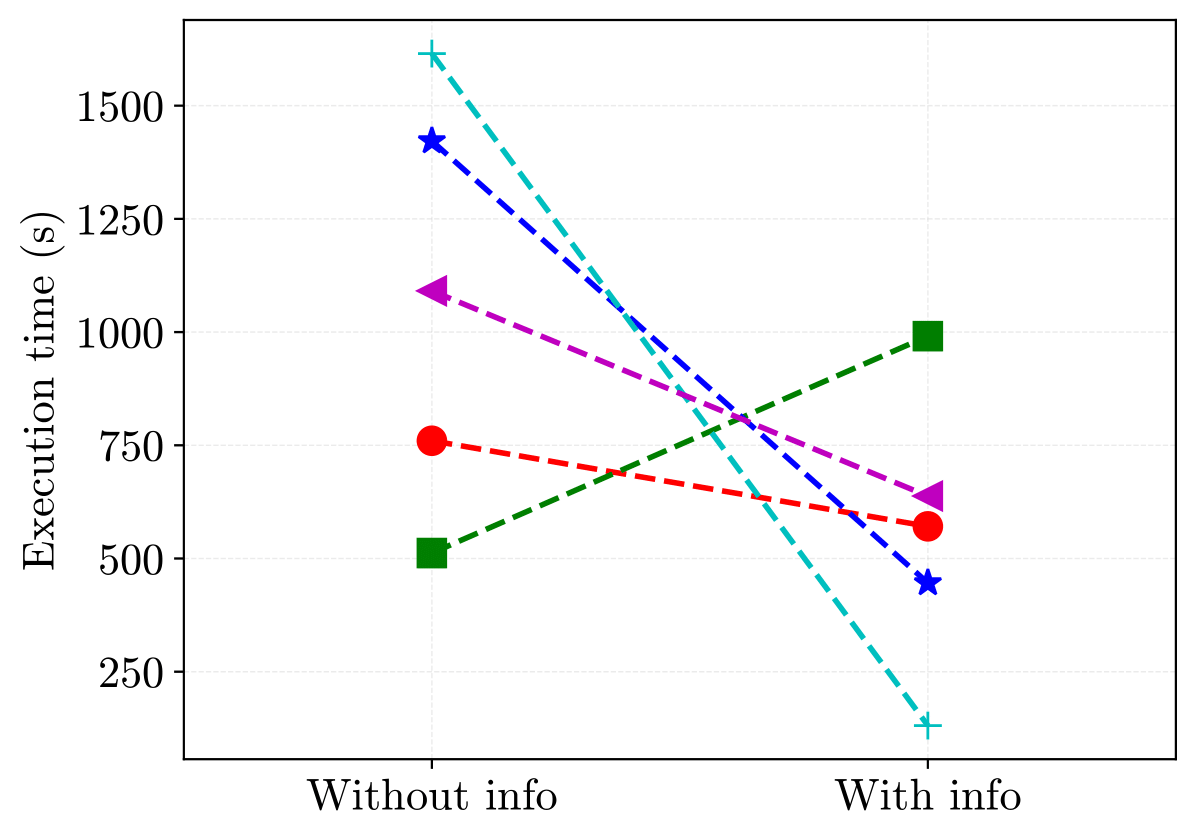} 
    \caption{\label{fig:runtime} The execution time of the model checker with and without the information of stubborn agents.}
    \vspace{-0.50cm}
\end{figure}

\section{Conclusion}\label{sec:conclusion}
We considered a verification problem for the FJ model with binary observations.
We constructed finite abstractions of the FJ model and showed that these abstractions can approximately simulate the original FJ dynamics.
Future work includes optimizing the verification implementation and applying the framework to datasets.







\bibliographystyle{ieeetr}
\bibliography{bibliography}

\end{document}